\pdfoutput=1
\documentclass[centertags,12pt]{amsart}

\usepackage[foot]{amsaddr}
\usepackage{amssymb}
\usepackage{amsmath}
\usepackage{bm}
\usepackage{amsthm}
\usepackage{color}

\usepackage[hidelinks]{hyperref}

\usepackage[stretch=10,shrink=10]{microtype}

\makeatletter
\renewcommand{\subsection}{\@startsection
{subsection}{2}{0mm}{\baselineskip}{-0.25cm}
{\normalfont\normalsize\em}}
\makeatother

\newtheorem{theorem}{Theorem}
\newtheorem{proposition}[theorem]{Proposition}
\newtheorem{corollary}[theorem]{Corollary}
\newtheorem{lemma}[theorem]{Lemma}
{\theoremstyle{definition}
\newtheorem{definition}{Definition}
\newtheorem{example}{Example}
\theoremstyle{remark}
\newtheorem{remark}{Remark}

\DeclareMathOperator{\Supp}{supp}
\newcommand{\illum}[2]{\rho_{#1}(#2)}

\usepackage{algorithm}
\usepackage{algorithmic}

\begin{document}

\title[Algorithmic Approach to EAQECCs from the Hermitian Curve]{An Algorithmic Approach to Entanglement-Assisted Quantum Error-Correcting Codes from the Hermitian Curve}

\author{R. B. Christensen$^1$}
\address{$^1$Department of Mathematical Sciences, Aalborg University, Skjernvej 4A, 9220 Aalborg Øst, Denmark. Orcid 0000-0002-9209-3739}
\email{rene@math.aau.dk}

\author{C. Munuera$^2$}
\address{$^2$IMUVA-Mathematics Research Institute, Universidad de Valladolid, Paseo Bel\'en 7, 47011 Valladolid, Spain. Orcid 0000-0001-8386-3060}
\email{cmunuera@uva.es}

\author{F. R. F. Pereira$^3$}
\address{$^3$Department of Electrical Engineering, Federal University of Campina Grande, Rua Apr\'igio Veloso 882, 58190-970, Campina Grande, Para\'iba, Brazil. School of Science and Technology, University of Camerino, I-62032 Camerino, Italy. INFN, Sezione di Perugia, Via A. Pascoli, 06123 Perugia, Italy. Orcid 0000-0001-5638-6334}
\email{revson.ee@gmail.com}

\author{D. Ruano$^4$}
\address{$^4$IMUVA-Mathematics Research Institute, Universidad de Valladolid, Paseo Bel\'en 7, 47011 Valladolid, Spain. Orcid 0000-0001-7304-0087}
\email{diego.ruano@uva.es}

\begin{abstract}
We study entanglement-assisted quantum error-correcting codes (EAQECCs) arising from  classical one-point algebraic geometry codes  from the Hermitian curve with respect to the Hermitian inner product. Their only unknown parameter is $c$, the number of required maximally entangled quantum states since the Hermitian dual of an AG code is unknown. In this article, we present an efficient algorithmic approach for computing $c$ for this family of EAQECCs. As a result, this algorithm allows us to provide EAQECCs with excellent parameters over any field size.
\end{abstract}

\maketitle

{\small
    {\em Funding:} 
    This work was supported in part by Grant PGC2018-096446-B-C21 funded by\\MCIN/AEI/10.13039/501100011033 and by ``ERDF A way of making Europe'', by Grant RYC-2016-20208 funded by MCIN/AEI/10.13039/501100011033 and by ``ESF Investing in your future'', and by the European Union's Horizon 2020 research and innovation programme, under grant agreement QUARTET No 862644.
}

{
    \renewcommand\thefootnote{}
    \footnote{
        \rule{0em}{1.8em}\hspace{-\parindent}\ignorespacesafterend
This is a pre-copy-editing, author-produced PDF of an article accepted for publication in `Advances in Mathematics of Communications' following peer review. The definitive publisher-authenticated version is available online at: \url{https://doi.org/10.3934/amc.2021072}}
    \addtocounter{footnote}{-1}
}

\section{Introduction}
\label{section1}

Quantum error correcting codes (QECCs) can be constructed by means of classical linear codes with the use of the CSS method~\cite{Calderbank,Gottesman,Ketkar}. QECCs over $\mathbb{F}_q$, the finite field with $q$ elements, are constructed from self-orthogonal classical linear codes over $\mathbb{F}_q$ if we consider the Euclidean metric, and over $\mathbb{F}_{q^2}$ if we consider the Hermitian metric. In general the Hermitian metric produces QECCs with better parameters since we may consider linear codes over a bigger field, the downside being, however, that the computations are more involved.

Brun et al. proposed in \cite{Brun} a pre-sharing entanglement protocol between encoder and decoder which increases the communication capacity and simplifies the theory of  quantum error-correction since it is not restricted to self-orthogonal classical linear codes. These quantum codes are known as  {\em entanglement-assisted quantum error-correcting codes} (EAQECCs). They were introduced over the binary field and then extended for an arbitrary finite field in \cite{Galindo:2019}. As for QECCs, we may consider classical linear codes over $\mathbb{F}_q$ if we consider the Euclidean metric, and over $\mathbb{F}_{q^2}$ if we consider the Hermitian metric. EAQECCs have four fundamental parameters $[[n,k,d;c]]$: length, dimension, minimum distance and the number of required pairs of maximally entangled quantum states, respectively. 
  
Among the classical codes used to construct quantum codes by means of the CSS construction, codes arising from algebraic geometry resources and tools (the so-called algebraic geometry codes or AG codes for short) have revealed themselves as good candidates to obtain quantum codes with good parameters \cite{Ashikhmin,Bartoli,Kim,Matsumoto}. In particular, we should highlight AG codes from Castle curves (the so-called Castle codes) \cite{Fer,MTT}, whose prototype is the Hermitian curve, which has traditionally played an important role in the theory of AG codes. Therefore, it is natural to consider the Hermitian curve to obtain EAQECCs with good parameters.

EAQECCs from AG codes have been addressed in \cite{Francisco1,Francisco2}, where the Euclidean inner product is mainly considered. In this work, we consider Hermitian codes \cite{HLP} -- that is, one-point AG codes from the Hermitian curve over $\mathbb{F}_{q^2}$ -- to obtain EAQECCs considering the Hermitian inner product. The length, dimension, and minimum distance of the EAQECCs from the Hermitian curve are known. Only one parameter remains: the number of required maximally entangled quantum states, $c$. Unlike the other parameters, however, $c$ is unknown if we consider the Hermitian inner product. Thus, the main task of this article is to compute $c$.

We may compute $c$ as the rank of a product of generator matrices, specifically the generator matrix of the classical linear code and the generator matrix of its Hermitian dual  (see \cite[Proposition 3]{Galindo:2019}). However, the Hermitian dual of an AG code (and of a Hermitian code) is unknown in general and, hence, there is no formula for computing $c$ in the Hermitian metric case. In this work, we present an algorithmic approach for computing $c$. The algorithm is based on a careful analysis of reduced polynomials modulo the Hermitian curve equation and on the valuation given by the pole order at the infinity point. This allows us to compute the parameters of EAQECCs from the Hermitian curve and illustrate that they have excellent parameters. The computational complexity of our algorithm is $\mathcal{O}(q^5)$, which is significantly better than computing a rank of a matrix using methods from linear algebra since the length of a Hermitian code is $q^3$. 

Finally, by considering $q$-adic expansions and Lucas' theorem, we are able to reduce the number of iterations in our algorithm by a careful analysis of the $q$-th powers of Hermitian codes. The proofs are different for $q$ prime and non-prime, but they are developed in a parallel way. Naturally, the proof for the non-prime case is more involved and technical.

The structure of this paper is as follows. In Section \ref{section2}, we present some standard results on classical and quantum codes. 
In Section \ref{Sect3}, we review Hermitian codes and introduce the idea of reduced polynomials that is used in the following section. 
Section \ref{Sect4} shows the first main result: we describe an algorithmic approach to the problem of determining the parameters of EAQECCs derived from a one-point Hermitian code. Moreover, we provide codes with excellent parameters at the end of the section. Finally, in Section \ref{Sect5}, we analyze the $q$-th powers of Hermitian codes, which allow us to improve the number of iterations of the algorithm in Section \ref{Sect4}. We also compute the computational complexity of our algorithm and show, computationally, that the parameters of most of the EAQECCs in this family exceed the Gilbert-Varshamov bound \cite{Galindo:2019}.
 
The ideas and methods developed in this article can serve as a paradigm to  study EAQECCs coming from of other AG codes obtained from different curves, and especially from Castle-type curves.

\section{Preliminaries}
\label{section2}

In this section we recall some preliminary and well known results about linear vector spaces over a finite field, the Frobenius map and entan\-glement-assisted quantum codes. 
A more complete study of these topics can be found in \cite{Galindo:2019,LaGuardia,Ketkar,LidlNiederreiter:Book}.

\subsection{The Frobenius map and the Hermitian dual of a subspace}
\label{section2.1}

Let $q=p^r$ be a prime power, and let $\mathbb{F}_{q^2}$ be the finite field with $q^2$ elements.
Additionally, let $\sigma:\mathbb{F}_{q^2}\rightarrow\mathbb{F}_{q^2}$ be the Frobenius map, $\sigma(x)=x^q$.
It is well known that $\sigma$ is an $\mathbb{F}_q$-automorphism of $\mathbb{F}_{q^2}$ and an involution,
that is $\sigma^2=id$. Furthermore, considering the same notation, $\sigma$ may be extended to the map
$\sigma:\mathbb{F}_{q^2}^n\rightarrow\mathbb{F}_{q^2}^n$, by $\sigma(\mathbf{v})=(v_1^q,\dots,v_n^q)$.
This map is $\sigma$-semilinear and bijective, thus it preserves  intersections and linear dependence: given
$V,W\subset \mathbb{F}_{q^2}^n$, we have  $\sigma(V\cap W)=\sigma(V)\cap\sigma(W)$ and
$\dim (\langle V\rangle)=\dim (\langle \sigma(V)\rangle)$, where $\langle V\rangle$
denotes the $\mathbb{F}_{q^2}$-linear subspace spanned by $V$. Following the usual notation in linear coding theory,
we often write $\mathbf{v}^q=\sigma(\mathbf{v})$ and $V^q=\sigma(V)=\{\mathbf{v}^q \ : \ \mathbf{v}\in V\}$.

The Euclidean and Hermitian inner products in $\mathbb{F}_{q^2}^n$ are defined as follows:
given $\mathbf{v}, \mathbf{w}\in \mathbb{F}_{q^2}^n$,
\[
  \mathbf{v} \cdot \mathbf{w}= v_1w_1+\cdots+ v_nw_n \; \mbox{ and } \;
  \mathbf{v} \cdot_H \mathbf{w}=\sigma(\mathbf{v}) \cdot \mathbf{w}.
\]
Given a linear subspace $V\subset \mathbb{F}_{q^2}^n$, its Euclidean and Hermitian duals are, respectively,
\begin{eqnarray*}
V^{\perp}&=&\{\mathbf{x}\in\mathbb{F}_{q^2}^n \ : \ \mathbf{v} \cdot \mathbf{x}=0 \mbox{ for all $\mathbf{v}\in V$} \}, \\
V^{\perp_H}&=&\{\mathbf{x}\in\mathbb{F}_{q^2}^n \ : \ \mathbf{v} \cdot_H \mathbf{x}=0 \mbox{ for all $\mathbf{v}\in V$} \}.
\end{eqnarray*}
We define $\Delta(V)=\dim(V\cap V^{\perp_H})$.

\begin{proposition} \label{propDelta(V)}
Let $V\subset \mathbb{F}_{q^2}^n$ be a linear subspace. We have \newline
(a) $V^{\perp_H}=(V^q)^{\perp}=(V^{\perp})^q$;\newline
(b) $\Delta(V)=\dim (V^q\cap V^{\perp})$; \newline
(c) $\Delta(V)=\Delta(V^{\perp})$.
\end{proposition}
\begin{proof}
(a) The first equality is clear. The second one follows from the fact that
$\mathbf{v}^q \cdot \mathbf{x}=0$ if and only if $\mathbf{v} \cdot \mathbf{x}^q=0$. Write $\mathbf{y}=\mathbf{x}^q$.
Then $V^{\perp_H}=\{\mathbf{y}^q\in\mathbb{F}_{q^2}^n \ : \ \mathbf{v} \cdot \mathbf{y}=0 \mbox{ for all $\mathbf{v}\in V$} \}=(V^{\perp})^q$,
and the second equality of (a) holds. (b) Applying $\sigma$ we have
$\Delta(V)=\dim(V\cap (V^{\perp})^q)=\dim(V^q \cap V^{\perp})$ which proves the statement in (b).
(c) Since $(V^{\perp})^{\perp}=V$ we get
$\Delta(V^{\perp})=\dim((V^{\perp})^q\cap (V^{\perp})^{\perp})=\Delta(V)$.
\end{proof}

In what follows, we use the expression given by Proposition \ref{propDelta(V)}(b) to compute $\Delta(V)$.

\subsection{Entanglement-Assisted Quantum Error-Correcting Codes}
\label{Sect2.2}

A quantum code $\mathcal{Q}$ is called an $[[n,k,d;c]]_q$ {\em entanglement-assisted quantum error-correcting code} (EAQECC) if it encodes $k$ logical qudits into $n$ physical qudits using $c$ copies of maximally entangled quantum  states
and it can correct $\lfloor(d-1)/2\rfloor$ quantum errors. The {\em rate} of an EAQECC is  $k/n$, its {\em relative distance} is $d/n$, and its {\em entanglement-assisted rate} is  $c/n$. When  $c = n-k$ then the EAQECC is said to have {\em maximal entanglement}.

\begin{proposition}\cite[Proposition 3 and Corollary 1]{Galindo:2019} \label{Prep:WildeHerm}
Let $\mathcal{C}$ be a linear code over $\mathbb{F}_{q^2}$ with parameters $[n,k,d]_{q^2}$.
Then there is an EAQECC $\mathcal{Q}$ with parameters
$[[n,2k-n+c, d'; c]]_q$, where $d'$
is the minimum Hamming weight of a vector in the set $\mathcal{C}\setminus(\mathcal{C}\cap \mathcal{C}^{\perp_H})$, and
\begin{equation}\label{formulac}
  c = \dim \mathcal{C}^{\perp_H} - \dim (\mathcal{C}^{\perp_H}\cap \mathcal{C})
\end{equation}
is the number of required maximally entangled quantum states.
\end{proposition}

From Propositions~\ref{propDelta(V)} and \ref{Prep:WildeHerm}, the amount of entanglement of a code $\mathcal{Q}$ given by Eq. (\ref{formulac}) can be computed as 
$c = \dim (\mathcal{C^\perp}) - \Delta(\mathcal{C})=n-\dim (\mathcal{C}) - \Delta(\mathcal{C})$. In this paper we are going to use this
formula to determine $c$. Furthermore, we clearly have that $d'\ge d(\mathcal{C})$.

In order to show the performance of the EAQECCs given in this article we will consider the Gilbert-Varshamov bound \cite[Theorem 5]{Galindo:2019} for EAQECCs considering the Hermitian inner product. 
 
\begin{theorem}\label{thmgv}
Assume the existence of positive integers $n$, $k\leq n$, $d$, $c\leq (n-k)/2$
such that
  \begin{equation}
    \frac{q^{n+k}-q^{n-k-2c}}{q^{2n}-1}\sum_{i=1}^{d - 1}
         {n \choose i}(q^2-1)^i < 1. \label{eqgv}
  \end{equation}
Then there exists an EAQECC $\mathcal{Q}$ with parameters
$[[n,k-c, d; c]]_q$.
\end{theorem}

Furthermore, we will consider the entanglement-assisted analog of the Singleton bound given in \cite[Theorem A.5]{Singleton}: if there exists an EAQECC with parameters $[[n,k,d;c]]_q$ and $d \le (n+2)/2$, then $2d \le n +2 - k  +c$. Note that the previous Singleton type bound does not hold for all range of parameters of an EAQECC. In the light of this bound, the following defect has been used to evaluate the performance of EAQECCs \cite[Notation 3.7]{Singleton}.

\begin{definition}
Let $\mathcal{Q}$ be an $[[n, k, d; c]]_q$ EAQECC. Then the  {\em entanglement-assisted quantum Singleton defect of $\mathcal{Q}$} is the integer \begin{equation}n + 2 - k + c - 2d.\end{equation} 
\end{definition}

Note that the entanglement-assisted quantum Singleton defect is non-negative if $d \le (n+2)/2$, and otherwise it may be negative. 

\section{Reduced polynomials for Hermitian codes}
\label{Sect3}

In this section we recall some basic facts about the Hermitian curve, the classical codes derived from it, Hermitian codes, and some tools we use in subsequent computations. For a complete treatment see \cite{HLP}.

\subsection{Hermitian curves and codes}
\label{Sect3.1}

Let $\mathcal{X}$ be the Hermitian curve defined over $\mathbb{F}_{q^2}$ by the affine equation
$x^{q+1} = y^q + y$. This is a nonsingular plane curve of genus $g=q(q-1)/2$,  with $n = q^3$
rational affine points, $P_1, \ldots, P_n$, plus one point at infinity, $Q$.
Let $\nu_Q$ be the valuation of $\mathbb{F}_{q^2}(\mathcal{X})$ given by the order at $Q$ and $\nu=-\nu_Q$.
From the equation of $\mathcal{X}$ we have $\nu(x)=q$ and $\nu(y)=q+1$.  
The following property of $\nu$  will be widely used in what follows.

\begin{proposition}\cite[Theorem 2.16(iii)]{HLP}\label{Lemma_0}
Given functions $f_1,f_2\in\mathbb{F}_{q^2}(\mathcal{X})^*$, we have $\nu(f_1+f_2)\leq \max\{\nu(f_1),\nu(f_2)\}$
with equality if $\nu(f_1)\neq\nu(f_2)$.
\end{proposition}

For $m\geq 0$, we consider the Riemann-Roch space
\[
  \mathcal{L}(m) = \mathcal{L}(mQ) = \{f\in \mathbb{F}_{q^2}(\mathcal{X}) : f = 0 \text{ or  div}(f)\ge -mQ  \}.
\]
This is a linear space, whose dimension is denoted by $\ell(m)$. Recall that $\ell(m) = m+1-g$ when
$m\geq 2g-1$, according to the Riemann-Roch theorem~\cite{HLP}.
Let $\mathcal{L}(\infty) = \cup_{m=0}^\infty \mathcal{L}(m)$ and denote by $S$ the Weierstrass semigroup of $Q$, then $S=\{ \nu(f) : f\in \mathcal{L}(\infty), f\neq 0 \}= \langle q, q+1\rangle$.
Since $\mathcal{L}(\infty)$ is a finitely generated $\mathbb{F}_{q^2}$-algebra, we have that
$\mathcal{L}(\infty) = \mathbb{F}_{q^2}[x,y]$. Abusing the notation, the elements of $\mathcal{L}(\infty)$ will be called {\em polynomials}.

\begin{lemma} \label{Lemma_2}
If $f_1, \dots, f_t\in \mathcal{L}(\infty)$ are nonzero polynomials of pairwise different orders,
then they are linearly independent.
\end{lemma}
\begin{proof}
Suppose $\lambda_1 f_1 + \cdots + \lambda_t f_t = 0$, then Proposition \ref{Lemma_0} ensures that
$\nu(0) = \nu(\lambda_1 f_1 + \cdots + \lambda_t f_t) = \max\{\nu(f_i) : \lambda_i\neq 0\}$. Hence
$\lambda_1 = \cdots = \lambda_t = 0$.
\end{proof}

Consider the evaluation map $ev:\mathcal{L}(\infty) \rightarrow \mathbb{F}_{q^2}^n$,
$ev(f)=(f(P_1),\dots,f(P_n))$. The algebraic geometry code $\mathcal{C}(m)$ is defined as
$\mathcal{C}(m) = ev(\mathcal{L}(m))$.  Since $\mathcal{C}(n+2g-1) = \mathbb{F}_{q^2}^n$ -- see eg.~\cite{Stichtenoth} --
we can restrict to $0\leq m \leq n+2g-1$. In this range, we can consider the quantum code $\mathcal{Q}(m)$
over $\mathbb{C}^q$ obtained from $\mathcal{C}(m)$ by the construction described in Proposition \ref{Prep:WildeHerm}. The code $\mathcal{Q}(m)$ has parameters $[[n,2k-n+c,\ge d;c]]_q$, where $n=q^3,k=k(m)$ and $d=d(m)$ are, respectively, the length,
dimension and minimum distance of $\mathcal{C}(m)$, and $c= n-k(m)-\Delta (\mathcal{C}(m))$ is the minimum number of maximally entangled quantum states 
consumed by the code. Recall that the dimension of $\mathcal{C}(m)$ is $k(m)=\ell(m)-\ell(m-n)$, and its exact minimum distance $d(m)$ has been determined in \cite{YK}.
Therefore, our main task is to compute $\Delta(m)=\Delta(\mathcal{C}(m))=\dim (\mathcal{C}(m)^q\cap \mathcal{C}(m)^{\perp})$, which is the only unknown value among the previous parameters. Please note that the dual code $\mathcal{C}(m)^{\perp}$ satisfies the well known relation $\mathcal{C}(m)^{\perp}=\mathcal{C}(m^{\perp})$, where $m^{\perp}=n+2g-2-m$, see \cite{HLP}.

\subsection{Reduced polynomials}
\label{Sect3.2}
A monomial $x^a y^b$ is called \emph{reduced} if $0\leq a<q^2$, $0\leq b<q$. Thus, the number of
reduced monomials is at most $n=q^3$. A polynomial $f\in\mathcal{L}(\infty)$ is \emph{reduced}  if $f=0$ or it is the sum of reduced monomials.  We denote by $\mathcal{R}$ the set of all reduced polynomials and
$\mathcal{R}^* = \mathcal{R}\setminus\{0\}$.  Consider the reduction
$\mathfrak{r}': \mathcal{L}(\infty) \rightarrow \mathcal{R}$ as follows. Given a polynomial $f$,
$\mathfrak{r}'(f)$ is obtained by performing the substitutions
\[
  \mbox{\bf (R1) } \mathfrak{r}'(y^q)= x^{q+1} - y \hspace*{5mm}  \mbox{ and }   \hspace*{5mm}  \mbox{\bf (R2) } \mathfrak{r}'(x^{q^2})= x
\]
as many times as possible. In other words, $\mathfrak{r}'(f)$ is the unique (up to multiplication by a constant $\lambda\in\mathbb{F}_{q^2}^*$) reduced polynomial in the coset of $f$ modulo the ideal $( x^{q+1}-y^q-y, x^{q^2} - x )\subset  \mathcal{L}(\infty)$. This comes from considering the ideal generated from the Hermitian equation, the field equations $x^{q^2}-x$ and $y^{q^2}-y$, and the lexicographical ordering with $y>x$.

The next lemma gives some properties of reduced polynomials that we shall use later.

\begin{lemma} \label{Lemma_1}
The following statements hold.\newline
(a) Two distinct reduced monomials  are linearly independent;\newline
(b) For any $f\in \mathcal{L}(\infty)$, we have $ev(f)=ev(\mathfrak{r}'(f))$;\newline
(c) The evaluation map $ev: \mathcal{R}\rightarrow\mathbb{F}_{q^2}^n$ is an isomorphism of vector spaces.
Thus, any code $\mathcal{C}\subseteq \mathbb{F}_{q^2}^n$ can be obtained by evaluating reduced polynomials.
\end{lemma}
\begin{proof}
(a) follows from the fact that $\nu(x^a y^b)=aq+b(q+1)$ and Lemma \ref{Lemma_2}. (b) is a consequence of $y^q= x^{q+1} - y$ in $ \mathcal{L}(\infty)$ and
$ev(x^{q^2})=ev(x)$. For (c), note that the order of a reduced polynomial is at most $q^3+q^2-q-1=n+2g-1$ and that $ev:\mathcal{L}(q^3+q^2-q-1)\rightarrow \mathbb{F}_{q^2}^n$ is surjective. From this, it follows that $ev: \mathcal{R}\rightarrow\mathbb{F}_{q^2}^n$ must be surjective as well.  Since the dimension of $\mathcal{R}$ is at most $n$, 
we get the isomorphism.
\end{proof}

In view of property (a) in Lemma \ref{Lemma_1}, any nonzero reduced polynomial has a unique leading monomial with respect to $\nu$. We define the normalization of $f\in\mathcal{R}$ as
$\mathfrak{n}(f)=0$ if $f=0$, and $\mathfrak{n}(f)=\lambda f$, where $\lambda\in\mathbb{F}_{q^2}$ is chosen
such that  $\mathfrak{n}(f)$ has leading coefficient equal to 1 (with respect to $\nu$), if $f\in\mathcal{R}^*$.
The normalized reduction of  $f\in \mathcal{L}(\infty)$ will be $\mathfrak{r}(f)=\mathfrak{n}(\mathfrak{r}'(f))$.

Let $\mathcal{M}=\{ 1,x,y,\dots\}=\{ f_1,f_2,\dots,f_n\}$ be the set of reduced monomials ordered
according to the values of $\nu$ and let  $\mathcal{M}_m=\{ f_1,\dots,f_{\ell(m)}\}$ for  $0\le m \le q^3+q^2-q-1$.  From Lemma \ref{Lemma_1} it holds that
$\{  ev(f_i) : f_i \in\mathcal{M}\}$ is a basis of $\mathcal{C}(q^3+q^2-q-1) \cong \mathbb{F}_{q^2}^n$ and
$\{  ev(f_i) : f_i \in\mathcal{M}_m\}$ is a basis of $\mathcal{C}(m)$.
In what follows, we deal with codes in terms of reduced polynomials.

\begin{example}
As mentioned before, the Hermitian curve over $\mathbb{F}_{q^2}$ has $q^3$ rational affine points and one point at infinity, which is also rational. The genus of the curve is $g = q(q-1)/2$.
    A basis for $\mathcal{L}(m)$ is given by the set $\{x^i y^j| i\geq 0, 0\leq j\leq q-1, iq+j(q+1)\leq m\}$. Now, assume that $q = 3$ and $m = 22$. 
    Then,
    \begin{equation*}
    B = \{1, x, y, x^2, xy, y^2, x^3, x^2y, xy^2, x^4, x^3y, x^2y^2, x^5, x^4y, x^3y^2,x^6, x^5y,x^4y^2, x^7, x^6y\}
  \end{equation*}
  is a basis for $\mathcal{L}(m)$. Consider $B^q = \{f^q\colon f\in B\}$. Since $q = 3$, the reductions are given by $\mathfrak{r}'(x^9) = x$ and $\mathfrak{r}'(y^3) = x^{4}-y$. Thus, the normalized reduction $\mathfrak{r}$ on $B^q$ gives the polynomials listed in Table~\ref{tab:reductions}.
    \begin{table}[h]
      \centering
      \setlength{\tabcolsep}{12pt}
      \begin{tabular}{rccc}
        $i$  & $f_i$    & $\mathfrak{r}'(f_i^q)$                   & $\mathfrak{r}(f_i^q)=\mathfrak{n}(\mathfrak{r}(f_i^q))$\\
        \hline
        $1$  & $1$      & $1$                           & $1$                           \\
        $2$  & $x$      & $x^3$                         & $x^3$                         \\
        $3$  & $y$      & $x^4$ $+$ $2y$                & $x^4$ $+$ $2y$                \\
        $4$  & $x^2$    & $x^6$                         & $x^6$                         \\
        $5$  & $xy$     & $x^7$ $+$ $2x^3y$             & $x^7$ $+$ $2x^3y$             \\
        $6$  & $y^2$    & $x^8$ $+$ $x^4y$ $+$ $y^2$    & $x^8$ $+$ $x^4y$ $+$ $y^2$    \\
        $7$  & $x^3$    & $x$                           & $x$                           \\
        $8$  & $x^2y$   & $2x^6y$ $+$ $x^2$             & $x^6y$ $+$ $2x^2$             \\
        $9$  & $xy^2$   & $x^7y$ $+$ $x^3y^2$ $+$ $x^3$ & $x^7y$ $+$ $x^3y^2$ $+$ $x^3$ \\
        $10$ & $x^4$    & $x^4$                         & $x^4$                         \\
        $11$ & $x^3y$   & $x^5$ $+$ $2xy$               & $x^5$ $+$ $2xy$               \\
        $12$ & $x^2y^2$ & $x^6y^2$ $+$ $x^6$ $+$ $x^2y$ & $x^6y^2$ $+$ $x^6$ $+$ $x^2y$ \\
        $13$ & $x^5$    & $x^7$                         & $x^7$                         \\
        $14$ & $x^4y$   & $x^8$ $+$ $2x^4y$             & $x^8$ $+$ $2x^4y$             \\
        $15$ & $x^3y^2$ & $x^5y$ $+$ $xy^2$ $+$ $x$     & $x^5y$ $+$ $xy^2$ $+$ $x$     \\
        $16$ & $x^6$    & $x^2$                         & $x^2$                         \\
        $17$ & $x^5y$   & $2x^7y$ $+$ $x^3$             & $x^7y$ $+$ $2x^3$             \\
        $18$ & $x^4y^2$ & $x^8y$ $+$ $x^4y^2$ $+$ $x^4$ & $x^8y$ $+$ $x^4y^2$ $+$ $x^4$ \\
        $19$ & $x^7$    & $x^5$                         & $x^5$                         \\
        $20$ & $x^6y$   & $x^6$ $+$ $2x^2y$             & $x^6$ $+$ $2x^2y$             \\
        \hline
      \end{tabular}
      \caption{Normalized reductions for $q=3$ and $m=22$}
      \label{tab:reductions}
    \end{table}
    Lastly, one can check that $\langle ev(B^q)\rangle = \langle ev(\mathfrak{r}(B^q))\rangle $.
\end{example}

\section{The parameters of EAQECCs from Hermitian codes}
\label{Sect4}

In this section we shall compute the parameters of the EAQECC $\mathcal{Q}(m)$ obtained  from the Hermitian code $\mathcal{C}(m)$. Let us remember that this problem leads us to the computation of
$\Delta(m)=\dim(\mathcal{C}(m)^{\perp_H} \cap \mathcal{C}(m))= \dim(\mathcal{C}(m)^{q} \cap \mathcal{C}(m^{\perp}))$, where $m^{\perp}=n+2g-2-m$.
In some cases, this number is easy to determine, see \cite{SK}.

\begin{proposition}  \label{Proposition_4}
If $m\leq q^2 - 2$, then $\Delta(m) = \ell(m)$.
\end{proposition}
\begin{proof}
Note that $\mathcal{L}(m)^q\subseteq \mathcal{L}(qm)$, where $\mathcal{L}(m)^q = \{f^q :  f\in\mathcal{L}(m)\}$.
If $m\leq q^2-2$, then $qm\leq m^\perp$ and thus $\mathcal{C}(m)^q\subseteq \mathcal{C}(m^\perp)$, so $\Delta(m)= \dim(\mathcal{C}(m)^q)=\dim(\mathcal{C}(m))$.
\end{proof}

For $m< q^3$ the evaluation map $ev:\mathcal{L}(m)^q\rightarrow \mathbb{F}_{q^2}^n$ is injective since $ev(f^q)=ev(f)^q$ and $ev(\mathcal{L}(m))$ is one-to-one. For $m\geq q^3$ this is no longer true. Thus, in this case we cannot expect a formula for $\Delta(m)$ in terms of
$\ell(m)$, as in Proposition~\ref{Proposition_4}. In the sequel, we develop a procedure that performs such a computation.
We can restrict to $m\geq q^2 -1$ since Proposition~\ref{Proposition_4} covers the remaining cases.
On the other hand, from Proposition \ref{propDelta(V)}(c), we can also assume $m\le m^{\perp}$, that is,
we can restrict to $m\le m^*=\lfloor n/2+g-1\rfloor$ \cite{Tiersma}. 
We  write $\mathcal{M}_*=\mathcal{M}_{m^*}$.

Let $m\le m^*$. The code $\mathcal{C}(m)^q$ can be obtained by
\[
\mathcal{C}(m)^q = ev(\langle f_1^q, \ldots, f_\ell^q\rangle) = ev(\langle \mathfrak{r}(f_1^q), \ldots, \mathfrak{r}(f_\ell^q)\rangle),
\]
where $\ell=\ell (m)$ (since $m<q^3=n$). Notice that the functions $\mathfrak{r}(f_1^q), \ldots, \mathfrak{r}(f_\ell^q)$ are linearly independent, as
$\dim(\mathcal{C}(m)^q)=\dim(\mathcal{C}(m))$. In general, however, they do not have pairwise different orders.
We shall construct a new set of functions $\phi_1,\dots,\phi_\ell$ such that
$\langle \phi_1, \ldots, \phi_\ell\rangle = \langle \mathfrak{r}(f_1^q), \ldots, \mathfrak{r}(f_\ell^q)\rangle$
and all $\phi_i$'s have different orders. To that end, we define a second reduction on the $\mathfrak{r}(f_i^q)$'s, denoted by $\mathfrak{s}$ and described iteratively via the rule
\[
  \mathfrak{s}(\mathfrak{r}(f_1^q)):=1
  \quad\text{and}\quad
  \mathfrak{s}(\mathfrak{r}(f_i^q)) := \mathfrak{r}(f_i^q)  \bmod \big\langle \mathfrak{s}(\mathfrak{r}(f_1^q)), \dots, \mathfrak{s}(\mathfrak{r}(f_{i-1}^q))\big\rangle, \; i=2,\dots,\ell .
\]
We then set $\phi_i=\mathfrak{s}(\mathfrak{r}(f_i^q))$ for $ i=1,\dots,\ell $.

This reduction may be described step-by-step as follows.
Define $\phi_1 := 1$. Since $f_1=1$, it is obvious that
$ \langle\phi_1\rangle=\langle\mathfrak{r}(f_1^q)\rangle$. Once normalized reductions
$\phi_1, \ldots, \phi_{t-1}$ such that  $\langle\phi_1, \dots, \phi_{t-1}\rangle=\langle \mathfrak{r}(f_1^q), \dots,  \mathfrak{r}(f_{t-1}^q) \rangle $ are computed, set $\phi =  \mathfrak{r}(f_t^q)$ and do:
\begin{description}
\item[(S1)] If $\nu(\phi) \neq \nu(\phi_i)$ for all  $i<t$, then set $\phi_t \leftarrow  \phi$. End.
\item[(S2)] If $\nu(\phi) = \nu(\phi_i)$ for some  $i<t$, then set $\phi \leftarrow \mathfrak{n}(\phi - \phi_i)$.
Repeat from (S1).
\end{description}
Note that when $\nu(\phi) = \nu(\phi_i)$ for some $i<t$, then we have  $\langle\phi_1, \dots, \phi_{t-1},\mathfrak{n}(\phi - \phi_i)\rangle=\langle\phi_1, \dots, \phi_{t-1},\phi \rangle$. Then, by the initial choice  $\phi =  \mathfrak{r}(f_t^q)$  and the induction hypothesis $\langle\phi_1, \dots, \phi_{t-1}\rangle=\langle \mathfrak{r}(f_1^q), \dots,  \mathfrak{r}(f_{t-1}^q) \rangle $, in each step of the previous procedure it holds that 
$\langle\phi_1, \dots, \phi_{t-1},\phi\rangle=\langle \mathfrak{r}(f_1^q), \dots,  \mathfrak{r}(f_t^q) \rangle$. In particular, the  functions $\mathfrak{r}(f_1^q), \ldots, \mathfrak{r}(f_\ell^q)$ being linearly independent, we deduce that $\phi\neq 0$. Note also that, in the case  $\nu(\phi) = \nu(\phi_i)$ for some $i<t$,
the leading terms of $\phi$ and $\phi_i$ coincide  by  Lemma~\ref{Lemma_1}. Hence, $\nu(\phi)$ decreases in each iteration, and after at most $t$ such
iterations, we obtain $\phi_t$.\footnote{We are showing in the following section that we can reduce the number of iterations. It is achieved by a careful analysis of the reduced monomials created in the present process.}
Furthermore, the $\phi_i$'s are reduced since they are linear combinations of reduced polynomials.
Thus, finally we obtain reduced functions  $\phi_1, \ldots, \phi_{t}$, of  pairwise different orders, such that
$\langle\phi_1, \ldots, \phi_{t} \rangle) = \langle \mathfrak{r}(f_1^q), \dots,  \mathfrak{r}(f_t^q) \rangle$.

Define the set $\Phi(m) = \{\phi_i  : 1\leq i\leq \ell(m) \}$. The above properties can be summarized as follows:

\begin{proposition} \label{Lemma_5}
$\Phi(m)$ is a set of linearly independent reduced polynomials with pairwise different orders and
 $\mathcal{C}(m)^q = ev(\langle\Phi(m)\rangle)$.
\end{proposition}

The set $\Phi(m)$ provides an efficient way to compute $\Delta(m)$. 

\begin{lemma} \label{Lemma_6}
Let $m, m'$ be two nonnegative integers with $m \leq m^*$. The set
$\{\phi\in\langle\Phi(m)\rangle : \nu(\phi)\leq m'\}$ is a linear space whose basis is
$\Phi(m,m')=\{\phi_i\in\Phi(m) : \nu(\phi_i)\leq m'\}$.
\end{lemma}
\begin{proof}
It is a consequence of Proposition \ref{Lemma_0} and Lemma \ref{Lemma_2}, taking into account
that all functions of $\Phi(m)$ have different orders.
\end{proof}

\begin{proposition} \label{Proposition_7}
Let $0\leq m \leq m^*$. We have $\Delta(m) = \#\{\phi_i\in\Phi(m) : \nu(\phi_i)\leq m^\perp\} = \#\Phi(m,m^\perp)$.
\end{proposition}
\begin{proof}
Let $\mathbf{x}\in\mathcal{C}(m)^q\cap\mathcal{C}(m^\perp)$. There exists a reduced polynomial
$f\in \mathcal{R}$ such that $\nu(f)\leq m^\perp$ and $ev(f) = \mathbf{x}$. Furthermore, by Proposition~\ref{Lemma_5},
there exists $\phi\in\langle\Phi(m)\rangle$ such that $\mathbf{x}=ev(\phi)$. Now, since both $f$ and $\phi$ are reduced, Lemma \ref{Lemma_1} gives $\phi = f$. Thus $\nu(\phi)\leq m^\perp$. Then the  number of vectors $\mathbf{x}$ in $\mathcal{C}(m)^q\cap\mathcal{C}(m^\perp)$ is exactly the number of $\phi$'s in $\mathcal{C}(m)^q$ such that $\nu(\phi)\leq m^{\perp}$.
The conclusion follows from Lemma~\ref{Lemma_6}.
\end{proof}

The previous arguments lead directly to an algorithmic way of computing $\Delta(m)$. Note that $-\nu$ is the valuation given by the point at infinity and $\mbox{LC}$ stands for leading coefficient (with respect to $\nu$).

\newcommand{\deccc}{\ensuremath{\mbox{\sc A basis for $\Delta(m)$ (for computing $c$)}}}
\begin{algorithm}[h!]
 \caption{\deccc}
 \algsetup{indent=2em}
 \begin{algorithmic}[1]\label{alg:dec1}
    \REQUIRE $m$ with  $q^2 -1 \le m \le m^*=\lfloor n/2+g-1\rfloor=\lfloor q^3/2+(q^2-q)/2-1\rfloor$;
                    $\mathcal{M}_m = \{f_1, \ldots, f_\ell \}$, where $\ell=m+1-(q^2-q)/2$.
    \ENSURE A basis for $\Delta(m)$.
    \medskip
    \FOR{$j=1, \ldots, \ell$}
         \STATE $\phi_j \gets f_j^q$;     
         \WHILE{$\deg_y (\phi_j)\ge q$ or $\deg_x (\phi_j) \ge q^2$}
             \STATE In $\phi_j$, substitute $y^q$ for $x^{q+1}-y$ and $x^{q^2}$ for $x$;          
         \ENDWHILE    
         \STATE $\phi_j \gets \phi_j/\mbox{LC}(\phi_j)$;
    \ENDFOR
    \FOR{$j=2, \ldots , \ell$}        
         \STATE $i=1$;
         \WHILE{$i<j$}
            \IF{$\nu(\phi_j) = \nu(\phi_i)$}
            	\STATE $\phi_j \gets \phi_j - \phi_i$;\label{alg:reduction1}
                \STATE $\phi_j \gets \phi_j/\mbox{LC}(\phi_j)$;\label{alg:reduction2}
                \STATE $i \gets 1$;
            \ELSE
            	\STATE $i \gets i + 1$;
            \ENDIF
         \ENDWHILE        
    \ENDFOR    
\STATE $\Delta(m)\gets\{\phi_1 , \ldots , \phi_\ell \}$;
\RETURN $\Delta(m)$;
\end{algorithmic}
\end{algorithm}

The next remark clarifies some properties of Algorithm~\ref{alg:dec1}.

\begin{remark}
(a) Note that all polynomials involved in the computations described above have coefficients in
$\mathbb{F}_p$, where $p$ is the characteristic of $\mathbb{F}_{q^2}$. Therefore, the algorithm runs over
$\mathbb{F}_p$, although both the Hermitian curve $\mathcal{X}$ and AG codes obtained from them
are defined over $\mathbb{F}_{q^2}$. \newline
(b) In view of Proposition \ref{propDelta(V)}(c), we have treated only the case  $m\le m^*$. For larger
values of $m$, one can use the identity $\Delta(m)=\Delta(m^{\perp})=\#\Phi(m^\perp,m)$. \newline
(c) For $m\ge n+2g-2-(q^2-2)=q^3-q$, we have $m^\perp\leq q^2-2$. By Proposition~\ref{propDelta(V)} and Proposition~\ref{Proposition_4} this implies $\Delta(m)=\Delta(m^\perp)=\ell(m^\perp)$. Hence, the EAQECC constructed from $\mathcal{C}(m)$ has entanglement $c=\ell(m^\perp)-\ell(m^\perp-n)-\Delta(m)=0$, meaning that it is a standard quantum code.
Similarly, if $m\le  q^2-2$ we have $\Delta(m)=\ell(m)$ and $\mathcal{C}(m)$ has dimension $k(m)=\ell(m)$. Thus, the resulting EAQECC has dimension $0$.
All these codes  can be discarded if one is only interested in EAQECCs with $c>0$ and $k>0$.
\end{remark}

\begin{example}
Let $q=3$. Here $g=3$ and  $m^*=15$. Algorithm~\ref{alg:dec1} gives the data shown in Table~\ref{tab:tabla1}.
\small
\begin{table}[h!]
\begin{center}
\begin{tabular}{cccccc} 
$f_i$    & $\nu(f_i)$ & $\mathfrak{r}(f_i^q)$ & $\nu(\mathfrak{r}(f_i^q))$ & $\phi_i$          & $\nu(\phi_i)$ \\ \hline
$1$      & 0          & $1$                   & 0                          & $1$               & 0   \\
$x$      & 3          & $x^3$                 & 9                          & $x^3$             & 9   \\
$y$      & 4          & $x^4+2y$              & 12                         & $x^4+2y$          & 12  \\
$x^2$    & 6          & $x^6$                 & 18                         & $x^6$             & 18  \\
$xy$     & 7          & $x^7+2x^3y$           & 21                         & $x^7+2x^3y$       & 21  \\
$y^2$    & 8          & $x^8+x^4y+y^2$        & 24                         & $x^8+x^4y+y^2$    & 24  \\
$x^3$    & 9          & $x$                   & 3                          & $x$               & 3   \\
$x^2y$   & 10         & $x^6y+2x^2$           & 22                         & $x^6y+2x^2$       & 22  \\
$xy^2$   & 11         & $x^7y+x^3y^2+x^3$     & 25                         & $x^7y+x^3y^2+x^3$ & 25  \\
$x^4$    & 12         & $x^4$                 & 12                          & $y$               & 4   \\
$x^3y$   & 13         & $x^5+2xy$             & 15                         & $x^5+2xy$         & 15  \\
$x^2y^2$ & 14         & $x^6y^2+x^6+x^2y$     & 26                         & $x^6y^2+x^6+x^2y$ & 26  \\
$x^5$    & 15         & $x^7$                 & 21                         & $x^3y$            & 13  \\
\hline
\end{tabular}
\caption{Algorithm~\ref{alg:dec1} for $q=3$.}
\label{tab:tabla1}  
\end{center}
\end{table}

In the first column we list the monomials in $\mathcal{M}_*$.
The second and sixth columns of this table allows us to compute $\Delta(m)$ for all values of $m$.
For example, $\Delta(10)=\#\Phi(10,21)=6$. If $m>15$ we apply the duality property; for example
$\Delta(21)=\#\Phi(21,10)=\#\Phi(10,21)=6$. 

We obtain quantum codes over $\mathbb{C}^3$ with parameters
$[[27, 1, 19; 16]]_3$, $[[27, 4, 16; 13]]_3$,  $[[27, 6, 13; 9]]_3$,  $[[27, 9, 10; 6]]_3$,  $[[27, 13, 7; 4]]_3$,
and $[[27, 16, 4; 1]]_3$. All previous codes have entanglement-assisted quantum Singleton defect equal to 6. An important feature of the first three examples above is 
that there is no (quantum) stabilizer code attaining the respective parameters.

Continuing the analysis for EAQECCs over $\mathbb{C}^4$, we derive the codes 
$[[64, 1, 49; 45]]_4$, $[[64, 5, 42; 35]]_4$, $[[64, 16, 30; 22]]_4$, 
$[[64, 24, 21; 12]]_4$, $[[64, 33, 14; 7]]_4$, $[[64, 35, 12; 3]]_4$, and\linebreak[4]
$[[64, 39, 8; 1]]_4$. They have entanglement-assisted quantum Singleton defect equal to
$12$, $12$, $12$, $12$, $12$, $10$, and $12$. 
Lastly, considering EAQECCs over $\mathbb{C}^5$, the codes have parameters equal to 
$[[125, 1, 101; 96]]_5$, $[[125, 9, 91; 84]]_5$, $[[125, 15, 81; 70]]_5$, $[[125, 36, 56; 41]]_5$, 
$[[125, 54, 41; 29]]_5$, $[[125, 70, 26; 15]]_5$, and $[[125, 90, 10; 1]]_5$, which have 
entanglement-assisted quantum Singleton quantum defect equal to 
$20, 20, 20, 20, 20, 20$, and $18$.  Regarding parameters attainability using stabilizer codes, we have that 
the first three and the first four EAQECCs over $\mathbb{C}^4$ and over $\mathbb{C}^5$, respectively, 
are unique. We summarize these results in Table \ref{tab:tabla2}, where GV stands for the Gilbert-Varshamov bound \cite{Galindo:2019}.

\begin{table}[h!]
\begin{center}
\begin{tabular}{ccccc} 
Parameters              & Singleton defect & Exceeding GV \\ \hline
$[[27, 1, 19; 16]]_3$   & 6                               & \checkmark                        \\
$[[27, 4, 16; 13]]_3$   & 6                               & \checkmark                        \\
$[[27, 13, 7; 4]]_3$    & 6                               & \checkmark                        \\
$[[27, 16, 4; 1]]_3$    & 6                               & \checkmark                        \\
$[[64, 5, 42; 35]]_4$   & 12                               & \checkmark                        \\
$[[64, 16, 30; 22]]_4$  & 12                               & \checkmark                        \\
$[[64, 35, 12; 3]]_4$   & 10                               & \checkmark                        \\
$[[64, 39, 8; 1]]_4$    & 12                               & \checkmark                        \\
$[[125, 1, 101; 96]]_5$ & 20                             & \checkmark                        \\
$[[125, 9, 91; 84]]_5$  & 20                             & \checkmark                        \\
$[[125, 36, 56; 41]]_5$ & 20                             & \checkmark                        \\
$[[125, 70, 26; 15]]_5$ & 20                             & \checkmark                        \\
 $[[125, 90, 10; 1]]_5$ & 18                               & \checkmark                        \\
\hline
\end{tabular}
\caption{Examples of code's parameters and comparative analysis by means of coding bounds.}\label{tab:tabla2}
\end{center}
\end{table}

\end{example}

\section{$q$-th powers of Hermitian codes}
\label{Sect5}

The computation of $\Delta(m)$ in the previous algorithm relies on the comparison of $\nu(f_i)$ and
$\nu(\phi_i)$ for all $f_i\in\mathcal{M}_*$.
The hardest part of this algorithm is the computation (and storage) of all reductions $\mathfrak{r}(f_i^q)$, $1\le i\le m^*$. 
However, in practice, the algorithm only requires the knowledge of $\mathfrak{r}(f_i^q)$ when $f_i$ is involved in some  reduction $\mathfrak{s}$, which does not happen for all the values of $i$ (see for example the case $ q = 3 $ in Table \ref{tab:tabla1}).  Otherwise (if $f_i$ is not involved in any reduction $\mathfrak{s}$), it is enough to know $\nu(\mathfrak{r}(f_i^q))$. In this section we show how this value can be determined directly, without computing $\mathfrak{r}(f_i^q)$ through {\bf (R1)} and  {\bf (R2)}.  
We start by giving a formula for the reduction  $\mathfrak{r}'(f_i^q)$ that allows us a fast computation of this data when necessary.

\subsection{Reducing $q$-th powers}
\label{Sect5.1}

Let us recall that the reduction $\mathfrak{r}'(f)$ of a polynomial $f\in \mathcal{L}(\infty)=\mathbb{F}_{q^2}[x,y]$  is obtained by performing the substitutions
\[
  \mbox{\bf (R1) } \mathfrak{r}'(y^q)= x^{q+1} - y \hspace*{5mm}  \mbox{ and }   \hspace*{5mm}  \mbox{\bf (R2) } \mathfrak{r}'(x^{q^2})= x
\]
as many times as possible. In other words, $\mathfrak{r}'(f)$ is the unique (up to multiplication by a constant $\lambda\in\mathbb{F}_{q^2}^*$) reduced polynomial in the coset of $f$ modulo the ideal $( x^{q+1}-y^q-y, x^{q^2} - x )\subset  \mathcal{L}(\infty)$, with respect to $\nu$. The next lemma gives some properties of $\mathfrak{r}'$.

\begin{lemma} \label{propreduc}
Let $g_1,g_2\in \mathbb{F}_{q^2}[x,y]$. The following properties hold. \newline
(a) $\mathfrak{r}'(g_1+g_2)=\mathfrak{r}'(g_1)+\mathfrak{r}'(g_2)$; \newline
(b) $\mathfrak{r}'(g_1g_2)=\mathfrak{r}'(\mathfrak{r}'(g_1) \mathfrak{r}'(g_2))$.
\end{lemma}
\begin{proof}
Both properties follow from the fact that the natural map $\mathbb{F}_{q^2}[x,y] \rightarrow \mathbb{F}_{q^2}[x,y]/( x^{q+1}-y^q-y, x^{q^2} - x )$ is a ring homomorphism. 
\end{proof}

As a notation, given  a non negative integer $\delta$, we write
\[
  \mathfrak{m}(\delta)= \left\{
    \begin{array}{ll}
      0                     & \mbox{if $\delta=0$}                         \\
      \delta \bmod{(q^2-1)} & \mbox{if $q^2-1\nmid \delta$}                \\ 
      q^2-1                 & \mbox{if $q^2-1|\delta$ and $\delta\neq 0$.} \\
    \end{array}
  \right.
\]
That is, for $\delta\neq 0$, $\mathfrak{m}(\delta)$ is the remainder of $\delta$ modulo $q^2-1$ in the interval $[1, q^2-1]$ rather than $[0,q^2-2]$ as usual. In particular   $\mathfrak{m}(\delta)\le q^2-1$. 

\begin{lemma}\label{sumarq^2+1}
Let $\delta_1,\delta_2$ be nonnegative integers.  The following properties hold. \newline
(a)  $\mathfrak{m}(\delta_1+\delta_2)=\mathfrak{m}(\mathfrak{m}(\delta_1)+\delta_2)=\mathfrak{m}(\mathfrak{m}(\delta_1)+\mathfrak{m}(\delta_2))$.\newline
(b) If $\delta_2>0$ then $\mathfrak{m}(q^2-1+\delta_2)=\mathfrak{m}(\delta_2)$.
\end{lemma}

Given nonnegative integers $\delta,\mu$, we denote the binomial coefficient modulo $p$ as
\[
  \binom{\mu}{\delta}_{\!\!  p} = \binom{\mu}{\delta} \;  \bmod{p}
\]
where $p$ is the characteristic of $\mathbb{F}_q$.

\begin{proposition}\label{mathfrakr'}
Let $f=x^ay^b$ with $0\le a$ and $0\le b<q$.  Then we have
\[
  \mathfrak{r}'(f^q)=\sum_{j=0}^b (-1)^j\binom{b}{j}_{\!\! p} x^{\mathfrak{m}(\nu(f)-j(q+1))} y^j.
\]
\end{proposition}
\begin{proof}
If $b=0$ then only the reduction (R2) is involved in $\mathfrak{r}'$.  If $a=0$,  according to the Newton's binomial formula and Lemma \ref{propreduc}  we have
\begin{align*}
\mathfrak{r}'((y^b)^q)= \mathfrak{r}'((x^{q+1}-y)^b) &=\mathfrak{r}'\left(\sum_{j=0}^b (-1)^j\binom{b}{j}_{\!\! p} x^{(b-j)(q+1)} y^j\right)\\
                                                     &=\sum_{j=0}^b (-1)^j\binom{b}{j}_{\!\! p} x^{(b-j)(q+1)} y^j.
\end{align*}
If $ab>0$, according to Lemma \ref{propreduc} and the previous computations  we have
\begin{align*}
\mathfrak{r}'((x^ay^b)^q) &=\sum_{i=0}^b (-1)^j\binom{b}{j}_{\!\! p} \mathfrak{r}'(x^{\mathfrak{m}(aq)+(b-j)(q+1)}) y^j\\
&=\sum_{j=0}^b (-1)^j\binom{b}{i}_{\!\! p} x^{\mathfrak{m}(aq+(b-j)(q+1))} y^j
\end{align*}
since all summands in the last expression are reduced monomials.
\end{proof}

\subsection{Computing  $\nu(\mathfrak{r}(f_i^q))$. Case $q$ prime}
\label{Sect5.2}

From Proposition \ref{mathfrakr'} we can deduce the values $\nu(\mathfrak{r}(f_i^q))$, for $f_i\in \mathcal{M}_*$ we need to run our algorithm. 
When $q=p$ is a prime number, then all binomial coefficients in the formula of Proposition \ref{mathfrakr'} are  non-zero. When $q$ is not a prime then some of these binomial coefficients may vanish, and the description of $\nu(\mathfrak{r}(f_i^q))$ becomes more involved.  We first study the case in which $q$ is a prime number. The case $q$  not a prime will be treated in the next subsection.

Let $f\in \mathcal{M}_*$. We will write the order $\nu(f)$ as $uq^2+sq+t$ with $0\leq s,t<q$ and  $0\le u\leq (q+1)/2$.  Note that this representation is unique.

\begin{proposition} \label{prop:qthOrders}
Let $q=p$ be a prime, and let $f\in\mathcal{M}_*$. Write $\nu(f)=uq^2+sq+t$ with $0\leq u\leq (q+1)/2$ and $0\leq s,t<q$. Then
  \[
    \nu(\mathfrak{r}(f^q))=
    \begin{cases}
      sq^2+(u+t)q & \mbox{\rm if $s\ge t$ and $u+t<q^2-sq$;} \\
      q^3-2q^2+(u+t)q+1 & \mbox{\rm if $s\ge t$ and $u+t\ge q^2-sq$;} \\
      q^3-q^2+(u+t-1)q+s+1 & \mbox{\rm if $s< t$ and $u+t\le q+s+1$;} \\
      q^3-2q^2+(u+t-1)q+s+2 & \mbox{\rm if $s< t$ and $u+t> q+s+1$.}          
    \end{cases}
  \]
\end{proposition}
\begin{proof}
The assumption $\nu(f)=uq^2+sq+t$ implies $f=x^{uq+s-t}y^t$ and therefore 
\begin{equation} \label{prooforderfq}
\mathfrak{r}'(f^q)=\sum_{j=0}^t (-1)^j\binom{t}{j}_{\!\! p} x^{\mathfrak{m}(uq^2+sq+t-j(q+1))} y^j
\end{equation}
from Proposition \ref{mathfrakr'}. Since $t<q=p$, all binomial coefficients are nonzero modulo $p$, so it is enough to find the highest order among all the summands in the previous expression.
To do this, we use the following observation. According to Lemma \ref{sumarq^2+1} (b),  when the condition $(S): sq+u+t>j(q+1)$ holds, then  $\mathfrak{m}(uq^2+sq+t-j(q+1))=\mathfrak{m}(sq+t+u-j(q+1))$. In particular $(S)$ is always satisfied for $j=0$, and for all values of $j$ when $s\ge t$.
Let us distinguish three separate cases.\newline
{\em Case 1:} If $s\ge t$ and $sq+u+t<q^2$, as above we have $\mathfrak{m}(sq+u+t-j(q+1))=sq+u+t-j(q+1)$. The maximum order among the summands of $\mathfrak{r}'(f^q)$  is $sq^2+(u+t)q$, obtained for $j=0$. \newline
{\em Case 2:} If $s\ge t$ and $sq+u+t\ge q^2$, then $s=q-1, t\ge 1$ and $sq+u+t\le q^2+q$, hence  $\mathfrak{m}(sq+u+t-j(q+1))=sq+u+t-(q^2-1)$ if $j=0$ and $\mathfrak{m}(sq+u+t-j(q+1))=sq+u+t-j(q+1)$ if $j>0$. Thus the maximum order in $\mathfrak{r}'(f^q)$ is  $q^3-2q^2+(u+t)q+1$, which is obtained for $j=1$.
\newline
{\em Case 3:} If  $s< t$, then $u> 0$. Note that the condition $(S)$ is satisfied for all $j\le s$. Since $sq+u+t< q^2$,  in this range we have $\mathfrak{m}(sq+u+t-j(q+1))=sq+u+t-j(q+1)$. So  the maximum order among the summands of $\mathfrak{r}'(f^q)$ corresponding to $j\le s$ is $sq^2+(u+t)q$, which is obtained for $j=0$. Consider now the summand corresponding to $j=t$. Since $uq^2+sq+t-t(q+1)= (u-1)(q^2-1)+(q+s-t)q +u-1$, we have $\mathfrak{m}(uq^2+sq+t-t(q+1))=(q+s-t)q +u-1$ and the corresponding monomial has order $(q+s-t)q^2 +(u-1)q+t(q+1)$, which is bigger than the orders we have obtained for $j\le s$. 
Finally let $j=t-h>s$. According to Lemma \ref{sumarq^2+1}(a) and the computation made for $j=t$, we have
$\mathfrak{m}(uq^2+sq+t-j(q+1))=\mathfrak{m}((q+s-t)q+(u-1)+h(q+1))=\mathfrak{m}((q+s-t+h)q+(u+h-1))$,
hence the maximum order is obtained when $h$ is as large as possible satisfying the condition $(q+s-t+h)q+(u+h-1)\le q^2-1$, or equivalently $(t-s-h)q \ge u+h$. It is easy to check that such largest value of $h$ is 
$h=t-s-1$ (that is $j=s+1$) when $q \ge u+t-s-1$, and $h=t-s-2$ (that is $j=s+2$) when $q < u+t-s-1$. A straightforward computation gives the maximum order among the summands of  \eqref{prooforderfq} in both cases, which is $q^3-q^2+(u+t-1)q+s+1$ if  $q \ge u+t-s-1$ and $q^3-2q^2+(u+t-1)q+s+2$ if  $q< u+t-s-1$.
\end{proof}
Proposition~\ref{prop:qthOrders} gives us the $\nu$-value of the leading monomial in $f^q$, it can be used to determine when $\nu(\mathfrak{r}(f^q))= \nu(\mathfrak{r}(f'^q))$ holds for two functions $f, f'\in\mathcal{M}_*$. However, we will postpone this until after Proposition~\ref{prop:monomialSupport}, which concerns the $\nu$-values of \emph{all} monomials in $f^q$.
To state this proposition, we use $\Supp f$ to denote the monomial support of $f$. That is, if $f=\sum c_{a,b}x^ay^b\in\mathbb{F}_{q^2}[x,y]$, we define $\Supp f=\{x^ay^b\mid c_{a,b}\neq 0\}$.

\begin{proposition}\label{prop:monomialSupport}
Let $f\in\mathcal{M}$. Then
\begin{equation*}
\Supp \mathfrak{r}(f^q) \subseteq \left\{f_j\in \mathcal{M} : \nu(f_j)\equiv q\nu(f)\pmod{q^2-1}\right\}.
\end{equation*}
Furthermore, if $q+1$ divides $\nu(f)$, then $\nu(f)\equiv q\nu(f)\pmod{q^2-1}$.
\end{proposition}
\begin{proof}
  If $f=1$ then both results are clear. Let us assume $f\neq 1$.
  We have $\nu(f^q)=q\nu(f)$, so in particular, $\nu(f^q)\equiv q\nu(f)\pmod{q^2-1}$. Thus, it suffices to show that when applying each of the reductions \textbf{(R1)} and \textbf{(R2)} to $f^q$, the orders of all resulting monomials remain in the original equivalence class.

  Hence, let $f^q=x^ay^b$. If $a<q^2$ and $b<q$, then the result is immediate since $\mathfrak{r}(x^ay^b)=x^ay^b$. If $a\geq q^2$, we can apply \textbf{(R2)} once to obtain the monomial $x^{a-(q^2-1)}y^b$. This has order $\nu(x^ay^b)-q(q^2-1)$, meaning that the equivalence class modulo $q^2-1$ is preserved. If $b\geq q$, applying \textbf{(R1)} gives two monomials $x^{a+(q+1)}y^{b-q}$ and $x^ay^{b-(q-1)}$. These have orders $\nu(x^ay^b)$ and $\nu(x^ay^b)-(q^2-1)$, respectively. Again, the resulting orders are in the same equivalence class as $\nu(x^ay^b)$. This proves the first claim of the proposition.

  For the second claim, one has that $q+1\mid\nu(f)$ implies $(q-1)\nu(f)\equiv 0\pmod{q^2-1}$, which can be rearranged to obtain the result.
\end{proof}
\begin{corollary}\label{psi}
Let $f,f'\in\mathcal{M}_*$. If $\nu(\mathfrak{r}(f^q))= \nu(\mathfrak{r}(f'^q))$ then  $\nu(f)\equiv\nu(f') \pmod{q^2-1}$.
\end{corollary}
\begin{proof}
  Proposition~\ref{prop:monomialSupport} tells us that $\nu(\mathfrak{r}(f^q))\equiv q\nu(f)\pmod{q^2-1}$, and similarly for $\mathfrak{r}(f'^q)$. Combining these equivalences gives the result.
\end{proof}
We now have a necessary condition for two functions $f, f'\in\mathcal{M}_*$ to satisfy
$\nu(\mathfrak{r}(f^q))= \nu(\mathfrak{r}(f'^q))$. This can be used to bound the number of reductions in our algorithm in Section~\ref{Sect4}, since the reduction $\mathfrak{s}$ is only applied to the monomials $f_i$ for which there exists $f_j\in \mathcal{M}_*$ with $j<i$ and $\nu(\mathfrak{r}(f_j^q))= \nu(\mathfrak{r}(f_i^q))$. The number of such $f_j$'s should be moderate, as the following proposition shows that the map $f\mapsto\nu(f)\bmod{(q^2-1)}$ is quite uniformly distributed.

\begin{proposition} \label{papendice}
(a) The map $\mathcal{M}_*\rightarrow \mathbb{Z}/(q^2-1)\mathbb{Z}$ given by $f\mapsto \nu(f)\bmod{(q^2-1)}$ is surjective. \newline
(b) For $0\leq k<q^2-1$, the number of monomials $f\in\mathcal{M}$ with $\nu(f)\equiv k\pmod{q^2-1}$ is
\begin{itemize}
  \item $q+2$ when $k=0$
  \item $q+1$ when $k\neq 0$ and $q+1\mid k$
  \item $q$ when $k\neq 0$ and $q+1\nmid k$
\end{itemize}
\end{proposition}

The proof of this result can be found in the Appendix.
To illustrate the preceding results, we give the following example.
\begin{example}
  Let $q=5$, and consider $f=f_{24}=x^3y^3$ with $\nu(f_{24})=33=q^2+q+3$. Thus, $\nu(f_{24})\bmod{(q^2-1)}=9$.
  Turning to the $q$'th power, we see that $\mathfrak{r}(f_{24}^q)=3x^{21}y^2 + 4x^{15}y^3 + x^9 + 2x^3y$. Table~\ref{tab:exampleSupport} lists the orders of the monomials in $\Supp\mathfrak{r}(f_{24}^q)$ along with their orders modulo $q^2-1$. In each case, this remainder is $21$, which is exactly $q\nu(f_{24})\bmod{(q^2-1)}$.

  The remaining monomials $f_j$ with $\nu(f_j)\equiv 9\pmod{q^2-1}$ are $f_{48}=x^9y^2$, $f_{72}=x^{15}y$, $f_{96}=x^{21}$, and $f_{119}=x^{21}y^4$. Thus, there are a total of $5$ such monomials as predicted by Proposition~\ref{papendice}. If we consider the support of their $q$'th powers, we obtain the four monomials in Table~\ref{tab:exampleSupport} and the monomial $f_{60}=x^9y^4$.

  One may also note that there is a duality between these two sets of monomials. Namely, the monomial support of $\{f_{12}^q,f_{36}^q,f_{60}^q,f_{84}^q,f_{108}^q\}$ is exactly $\{f_{24},f_{48},\allowbreak f_{72},f_{96},f_{119}\}$.
\end{example}
\begin{table}[ht]
    \centering
    \begin{tabular}{cccc}
      $j$ & $f_j$ & $\nu(f_j)$ & $\nu(f_j)\bmod{(q^2-1)}$\\
      \hline
      $108$ & $x^{21}y^2$ & $117$ & $21$ \\
      $84$ & $x^{15}y^3$ & $93$ & $21$ \\
      $36$ & $x^9$ & $45$ & $21$ \\
      $12$ & $x^3y$ & $21$ & $21$ \\
      \hline
    \end{tabular}
    \caption{Monomials in the support of $\mathfrak{r}(f_{24}^q)$ for $q=5$}
	\label{tab:exampleSupport}
\end{table}

\subsection{Computing  $\nu(\mathfrak{r}(f_i^q))$. Case $q$ non-prime}
\label{sec:non-prime-q}

In this section we describe $\nu(\mathfrak{r}(f^q))$ when $q=p^r$ is not a prime, analogous to what was done in Proposition \ref{prop:qthOrders} when $q$ is a prime number. Our study will rely on Lucas' theorem, which relates the binomial coefficient $\binom{t}{j}\bmod{p}$ to the $p$-ary representations of $t$ and $j$. In \cite{GKS13}, the concept of a \emph{$p$-shadow} is defined as follows. Let $t=\sum_{i=0}^{r-1} t_i p^i$ and $j=\sum_{i=0}^{r-1} j_i p^i$ be the $p$-ary representations of $t$ and $j$, where $0\le j\le t<q$. If $j_i\leq t_i$ for all $i$, then $j$ is said to be in the $p$-shadow of $t$, and we write $j\leq_p t$. As a corollary to Lucas' theorem \cite{Fine47,lucas1891}, we  have
\begin{equation}\label{eq:lucasTheorem}
  \binom{t}{j}_{\!\!  p}\neq  0  \quad\mbox{ if and only if }\quad j\leq_p t.
\end{equation}
Inspired by the $p$-shadow, we introduce the \emph{$p$-illumination}.

\begin{definition}\label{defi:illumination}
Let $q=p^r$ be a prime power and let $0\leq j\leq t<q$ be integers. If $t=\sum_{i=0}^{r-1} t_i p^i$ and $j=\sum_{i=0}^{r-1} j_i p^i$ are the $p$-ary representations of $t$ and $j$, we let
\[
    \illum{t}{j}=\sum_{i=0}^{i^\ast} (j_i-t_i)p^i,
\]
where $i^\ast$ is the largest index such that $j_{i^\ast}>t_{i^\ast}$. If no such $i^\ast$ exists, we let $\illum{t}{j}=0$. We call $\illum{t}{j}$ the {\em $p$-illumination} of $j$ with respect to $t$. 
\end{definition}

It is easy to verify that $\illum{t}{j}$ is a non-negative integer. Additionally, $j-\illum{t}{j}$ is in the $p$-shadow of $t$, and it is the largest integer less than or equal to $j$ that satisfies this property. This follows from the observation that $j-\illum{t}{j}$ has the same $p$-ary digits as $t$ for indices $1,2,\ldots,i^\ast$. Each increment up to $j$ will have at least one of these digits greater than the corresponding $p$-digit of $t$. These considerations yield the following lemma.

\begin{lemma}\label{lem:binomials}
  Let $q=p^r$ be a prime power, $t<q$ and $0\leq j\leq t$. If $\illum{t}{j}=0$, then $\binom{t}{j}_{\!\!  p}\neq 0$. Otherwise, if $\illum{t}{j}>0$, then
  \[
    \binom{t}{j}_{\!\!  p}= \binom{t}{j-1}_{\!\!  p}= \cdots =\binom{t}{j-\illum{t}{j}+1}_{\!\!  p}= 0 \; 
    \mbox { and }
    \quad \binom{t}{j-\illum{t}{j}}_{\!\!  p}\neq 0.
  \]
\end{lemma}
\begin{proof}
  By \eqref{eq:lucasTheorem}, $\binom{t}{j'}_{\!\!  p}$ is non-zero if and only if $j'\leq_p t$. The observations immediately below Definition~\ref{defi:illumination} imply that  the first such $j'$ below $j$ is $j-\illum{t}{j}$, proving the lemma.
\end{proof}

In other  words, the largest integer $j^*\le j$ satisfying $\binom{t}{j^*}_{\!\!  p}\neq 0$ is $j^*=j-\rho_t(j)$. In addition, we can apply Lemma~\ref{lem:binomials} to $t-j$ and use the symmetry of the binomial coefficient, $\binom{t}{i}=\binom{t}{t-i}$, to infer that $\binom{t}{j}_{\!\! p}=\cdots=\binom{t}{j+\rho_t(t-j)-1}_{\!\! p}=0$ and $\binom{t}{j+\rho_t(t-j)}\neq 0$. That is, the smallest integer $j^*\ge j$ satisfying $\binom{t}{j^*}_{\!\!  p}\neq 0$ is $j^*=j+\rho_t(t-j)$.

\begin{proposition}\label{qthordersnoprime}
  Let $q=p^r$ be a prime power, and let $f\in\mathcal{M}_*$. Write $\nu(f)=uq^2+sq+t$ with $0\leq u\leq (q+1)/2$ and $0\leq s,t<q$. Then $ \nu(\mathfrak{r}(f^q))$ equals
\[
    \begin{cases}
      sq^2+(u+t)q & \mbox{\rm if $s\ge t$, $u+t<q^2-sq$;} \\
      q^3-2q^2+(u+t)q+1-(q^2-1)\rho_t(t-1) & \mbox{\rm if $s\ge t$, $u+t\ge q^2-sq$;} \\
      q^3-q^2+(u+t-1)q+s+1-(q^2-1)\rho_t(t-s-1) & \mbox{\rm if $s< t$, $u+t\le q+s+1$;} \\
      q^3-2q^2+(u+t-1)q+s+2-(q^2-1)\rho_t(t-s-2) & \mbox{\rm if $s< t$, $u+t> q+s+1$.}          
    \end{cases}
  \]  
\end{proposition}

\begin{proof}
The proof of this result is similar to that of Proposition \ref{prop:qthOrders}, in which we sought the summand $j$ providing the highest order in  the  writing of $\mathfrak{r}'(f^q)$ given in the Proposition \ref{mathfrakr'}.
Let $j^*$ be the index of the summand providing the maximum order in our case $q$ not a prime.
Following the proof of Proposition  \ref{prop:qthOrders}, in case~1 such maximum order is obtained for $j=0$. Since $\binom{t}{0}_{\!\!  p}\neq 0$, we have $j^*=0$.  In case~2, the maximum order comes from the summand corresponding to the smallest index $j^*\ge 1$ with $\binom{t}{j^*}_{\!\!  p}\neq 0$, that is for $j^*=1+\rho_t(t-1)$ according to Lemma \ref{lem:binomials}.
In case~3, the maximum order comes as well from the summand corresponding to the smallest index $j^*\ge j$ with $\binom{t}{j^*}_{\!\!  p}\neq 0$,  where $j=s+1$ when  $u+t\le q+s+1$ and $j=s+2$ when  $u+t> q+s+1$. Such indices are $j^*=s+1+\rho_t(t-s-1)$ and $j^*=s+2+\rho_t(t-s-2)$ respectively.
In all cases, it is enough to compute the orders of the summands corresponding to these $j^*$'s to obtain the stated formula.
\end{proof}

Note that when $\nu(f)\ge q^2$, the values  $ \nu(\mathfrak{r}(f^q))$ are, in most cases, bigger when $q$ is a prime number than when it is not, $q=p^r$ with $r>0$. Therefore, also $\Delta$ increases (and so the entanglement $c$ decreases) when $r$ increases.

\subsection{Bounding the complexity of Algorithm~\ref{alg:dec1}}
Having described the $q$'th powers given in the previous sections, we now use those results to bound the computational complexity of Algorithm~\ref{alg:dec1}. Since the most computationally costly part are the reductions $\mathfrak{s}$ -- that is, lines~\ref{alg:reduction1} and \ref{alg:reduction2} --  we focus on bounding the number of such reductions as well as bounding the cost of each reduction.

\begin{proposition}\label{prop:algComplexity}
  The total number of reductions $\mathfrak{s}$ in Algorithm~\ref{alg:dec1} is $\mathcal{O}(q^4)$, and each reduction requires $\mathcal{O}(q)$ field operations.
\end{proposition}
\begin{proof}
  First recall that $\phi_i=\mathfrak{r}(f_i^q)$. Thus, Proposition~\ref{prop:monomialSupport} ensures that all monomials $f_k$ in $\phi_i$ satisfy $\nu(f_k)\equiv q\nu(f_i)\pmod{q^2-1}$.
  If we consider some $\phi_j$ during the algorithm, similar arguments as in Corollary~\ref{psi} show that the reduction $\phi_j=\phi_j-\phi_i$ can happen only if $\nu(f_j)\equiv\nu(f_i)\pmod{q^2-1}$. Thus, we bound the total number of reductions by an amortized analysis, grouping polynomials $\phi_i$ based on the equivalence class of $\nu(f_i)$.

  Fix a $k\in\{0,1,\ldots,q^2-2\}$, and consider all polynomials $\phi_i$ such that $\nu(f_i)\equiv k\pmod{q^2-1}$. By Proposition~\ref{papendice}, there can be at most $q+2$ such monomials. The first time the algorithm encounters such a $\phi_i$ no reduction will happen. The second time such a $\phi_i$ is found, there will be at most one reduction. The third time at most two and so forth.
  Thus, the total number of reductions required to process the $\phi_i$ with $\nu(f_i)\equiv k\pmod{q^2-1}$ cannot exceed
  \[
    \sum_{j=1}^{q+2}(j-1)=\frac{q(q+1)}{2}.
  \]
  Because there are $q^2-1$ possible values of $k$, the overall number of reductions during the algorithm is at most $\frac{1}{2}q(q+1)(q^{2}-1)$, which is $\mathcal{O}(q^4)$.

  To prove the claim on the number of field operations per reduction, note that each $\phi_i$ can contain at most $q+2$ monomials by Propositions~\ref{prop:monomialSupport} and \ref{papendice}.
\end{proof}
Since the reductions $\mathfrak{s}$ are what dominates the complexity of Algorithm~\ref{alg:dec1}, Proposition~\ref{prop:algComplexity} implies that the number of field operations used during the algorithm is $\mathcal{O}(q^5)$, which is significantly better than computing $c$ as in Proposition \ref{Prep:WildeHerm} using linear algebra methods. Namely, it is significantly better than computing the rank of $\mathbf{G}\mathbf{G}^*$, where $\mathbf{G}$ is a generator matrix of the linear code and $\mathbf{G}^*$ is the $q$-th power of the transpose matrix of $\mathbf{G}$ (see \cite[Proposition 3]{Galindo:2019}).

We can optimize the algorithm even further by using the values of $\nu(\mathfrak{r}(f^q))$ found in Sections~\ref{Sect5.2} and \ref{sec:non-prime-q}.
Namely, the algorithm only requires the knowledge of $\mathfrak{r}(f_i^q)$ when $f_i$ is involved in some  reduction $\mathfrak{s}$, which does not happen for all the values of $i$ (see for example the case $ q = 3 $ in Table \ref{tab:tabla1}).

Therefore, we can modify the algorithm as follows: we first compute the orders $\nu(\mathfrak{r}(f_i^q))$ from Propositions  \ref{prop:qthOrders} and \ref{qthordersnoprime}. If $f_i$ is not involved in any  reduction $\mathfrak{s}$, then we will not calculate $\mathfrak{r}(f_i^q)$. Otherwise, when it is necessary to know this data, it is computed from Proposition~\ref{mathfrakr'}, and the corresponding polynomials $\phi$'s from the reduction $\mathfrak{s}$. 
The obtained result is shown in Table \ref{tab:tabla4}  for $q=3$.  The reader can compare this table with Table 
\ref{tab:tabla1}, in which the unmodified process is shown.

\begin{table}[h!]
\centering
\begin{tabular}{cccccc} 
$f_i$    & $\nu(f_i)$ & $\mathfrak{r}(f_i^q)$ & $\nu(\mathfrak{r}(f_i^q))$ & $\phi_i$    & $\nu(\phi_i)$ \\ \hline
$1$      & 0          &                       & 0                          &             & 0   \\
$x$      & 3          &                       & 9                          &             & 9   \\
$y$      & 4          &   $x^4+2y$       & 12                         &  $x^4+2y$    & 12  \\
$x^2$    & 6          &                       & 18                         &             & 18  \\    
$xy$     & 7          & $x^7+2x^3y$           & 21                    & $x^7+2x^3y$ & 21  \\
$y^2$    & 8          &                       & 24                         &             & 24  \\
$x^3$    & 9          &                       & 3                          &             & 3   \\
$x^2y$    & 10         &                       & 22                         &             & 22  \\
$xy^2$   & 11         &                       & 25                         &             & 25  \\
$x^4$    & 12         &  $x^4$                & 12                         &  $y$        & 4   \\
$x^3y$   & 13         &                       & 15                         &             & 15  \\
$x^2y^2$ & 14         &                       & 26                         &             & 26  \\
$x^5$    & 15         & $x^7$                 & 21                         & $x^3y$      & 13  \\
\hline
\end{tabular}
\caption{Results of the modified algorithm for $q=3$.}\label{tab:tabla4}
\end{table}

\subsection{Computational results}
\label{reults}
We have computed the parameters of all Hermitian EAQECCs up to field size $16$ using our algorithm. These lists of parameters reveal that many of the resulting EAQECCs exceed the Gilbert-Varshamov bound~\cite{Galindo:2019}. Once $\mathcal{C}(m)$ is large enough to produce an EAQECC of non-zero dimension, all following values of $m$ seem to produce codes exceeding the bound as well until the resulting entanglement $c$ is `too small' compared to the code length. More specifically, we have verified that the parameters of the EAQECCs exceed the Gilbert-Varshamov bound \cite{Galindo:2019} for the values of $c$ specified in Table~\ref{tab:tabla6}. For lower values of $c$, it is sometimes possible to find codes that exceed the bound, but it will not be true for all codes with this entanglement.
\begin{table}[h!]
  \centering
    \begin{tabular}{ccc}  
      $q$  & $n$ & $c$   \\ \hline
      2 & 8 & 0--3 \\
      3 & 27 & 1--16 \\
      4 & 64 &  3--45 \\
      5 & 125 & 4--96 \\
      7 & 343 & 10--288 \\
      8 & 512 & 9--441\\
      9 & 729 & 14--640 \\
      11 & 1331 & 38--1200 \\
      13 & 2197 & 51--2016\\
      16 & 4096 & 45--3825\\
      \hline
    \end{tabular}   
    \caption{Hermitian EAQECCs exceding the GV bound.}
    \label{tab:tabla6}
\end{table}

\section*{Acknowledgements}
The authors thank the anonymous reviewers, whose helpful comments have led to a better manuscript.

\bibliographystyle{plainurl}
\bibliography{ref.bib}

\appendix
\label{appendix}
\section{Proof of Proposition \ref{papendice}}

In this appendix we give the proof of  Proposition \ref{papendice}, which relies on the following lemmata.

\begin{lemma}\label{lem:largestElement}
  Let $f=x^ay^b\in\mathcal{M}$, and assume that any $f_j\in\mathcal{M}$ with $\nu(f_j)\equiv\nu(f)\pmod{q^2-1}$ satisfies $\nu(f_j)\leq\nu(f)$. Then $(q-2)(q+1)< a\leq (q-1)(q+1)$ and $0<b<q$.
\end{lemma}
\begin{proof}
  Assume first that $b=0$. Then $f'=x^ay^{q-1}$ satisfies $\nu(f')=\nu(f)+q^2-1$, meaning that $\nu(f')>\nu(f)$ and $\nu(f')\equiv\nu(f)\pmod{q^2-1}$. But this contradicts the choice of $f$. Thus, let $b>0$, and consider the case $a\leq (q-2)(q+1)$. As before, we can find $f'=x^{a+(q+1)}y^{b-1}\in\mathcal{M}$ satisfying $\nu(f')=\nu(f)+q^2-1$, which is again a contradiction.
\end{proof}

\begin{lemma}\label{lem:dividesXDeg}
  Let $f=x^ay^b\in\mathcal{M}$. Then $q+1\mid \nu(f)$ if and only if $q+1\mid a$.
\end{lemma}
\begin{proof}
  We have $\nu(f)=aq+b(q+1)$, which is divisible by $q+1$ if and only if $q+1\mid a$ because $q$ and $q+1$ are relatively prime.
\end{proof}

\renewcommand*{\thetheorem}{\ref{papendice}}

\begin{proposition} 
(a) The map $\mathcal{M}_*\rightarrow \mathbb{Z}/(q^2-1)\mathbb{Z}$ given by $f\mapsto \nu(f)\bmod{(q^2-1)}$ is surjective. \newline
(b) For $0\leq k<q^2-1$, the number of monomials $f\in\mathcal{M}$ with $\nu(f)\equiv k\pmod{q^2-1}$ is
\begin{itemize}
  \item $q+2$ when $k=0$
  \item $q+1$ when $k\neq 0$ and $q+1\mid k$
  \item $q$ when $k\neq 0$ and $q+1\nmid k$
\end{itemize}
\end{proposition}
\begin{proof}
  (a) is clear. For (b) we show how to construct the given number of monomials in each of the three cases. To see that there cannot be any additional monomials, note that there are $q-2$ values of $k$ such that $k\neq 0$ and $q+1\mid k$ and $q^2-q$ values such that $k\neq 0$ and $q+1\nmid k$. Hence, the total number of monomials considered is $(q+2)+(q-2)(q+1)+(q^2-q)q=q^3$ which is exactly the number of monomials in $\mathcal{M}$.

  Denote by $f^\ast=x^{a^\ast}y^{b^\ast}\in\mathcal{M}$ the monomial of largest order such that $\nu(f^\ast)\equiv k\pmod{q^2-1}$. Consider now the following two sequences of rational functions
  \begin{align}
    x^{a^\ast-i(q+1)}y^{b^\ast+i},&\quad i=0,1,\ldots,q-1-b^\ast\label{eq:firstSequence}\\
    x^{a^\ast-(q-1-b^\ast+i)(q+1)}y^i,&\quad i=0,1,\ldots,b^\ast.\label{eq:secondSequence}
  \end{align}
  The reader may note that all expressions in~\eqref{eq:firstSequence} and \eqref{eq:secondSequence} are distinct;
  setting $i=0$ in~\eqref{eq:firstSequence} gives $f^\ast$; and all $f$ in~\eqref{eq:firstSequence} and \eqref{eq:secondSequence} satisfy $\nu(f)\equiv\nu(f^\ast)\pmod{q^2-1}$. Thus, we only need determine which of these rational functions are in $\mathcal{M}$.

  Consider first an $f=x^ay^b$ from~\eqref{eq:firstSequence}. It is immediately clear that $0\leq b<q$, so $f\in\mathcal{M}$ if and only if $0\leq a<q^2$.
  That $a<q^2$ follows immediately from $a\leq a^\ast$ combined with the fact that $f^\ast\in\mathcal{M}$. To prove that $a$ is nonnegative, note that $f^\ast$ satisfies Lemma~\ref{lem:largestElement} by definition. This ensures that $a^\ast=(q-2)(q+1)+t$ for some $0<t\leq q+1$ and that $b^\ast>0$. The latter implies that $i\leq q-2$, and we obtain
  \[
    a =(q-2)(q+1)+t-i(q+1)
      \geq (q-2)(q+1)+t - (q-2)(q+1)
    =t>0.
  \]
  Hence, the rational functions in~\eqref{eq:firstSequence} are in $\mathcal{M}$ for all values of $k$.

  Considering now $f=x^ay^b$ from~\eqref{eq:secondSequence}, it is again straight-forward to check that $0\leq b <q$ and $a<a^\ast<q^2$ as above. Additionally, if $i\leq b^\ast-1$, we have
  \[
    a=(q-2)(q+1)+t-(q-1-b^\ast+i)(q+1)\geq (q-2)(q+1)+t-(q-2)(q+1)>0,
  \]
  so these are all in $\mathcal{M}$ regardless of the value of $k$. If $i=b^\ast$, however, we obtain $a=(q-2)(q+1)+t-(q-1)(q+1)$, which is only non-negative if $t=q+1$. This happens if and only if $q+1\mid a$, which in turn happens if and only if $q+1\mid k$ by Lemma~\ref{lem:dividesXDeg}.

  Summing up, \eqref{eq:firstSequence} and \eqref{eq:secondSequence} gives $(q-b^\ast)+b^\ast=q$ monomials if $q+1\nmid k$ and $(q-b^\ast)+(b^\ast+1)=q+1$ monomials if $q+1\mid k$.

  The final part of the proof is the following observation. The monomial $1$ is not contained in~\eqref{eq:firstSequence} or \eqref{eq:secondSequence} for any $k$. Thus, when $k=\nu(1)=0$, we obtain one extra monomial in addition to the $q+1$ from~\eqref{eq:firstSequence} and \eqref{eq:secondSequence}.
\end{proof}

\end{document}